\documentclass[smallabstract,smallcaptions]{dccpaper}

\usepackage{epsfig}
\usepackage{amsmath}
\usepackage{amssymb}
\usepackage{color}
\usepackage{url}

\newlength{\figurewidth}
\newlength{\smallfigurewidth}

\def\bfx{{ \bf x  }}
\def\bfX{{\bf X}}

\def\calX{{\mathcal{X}}}

\newtheorem{theorem}{Theorem}[section]
\newtheorem{lemma}[theorem]{Lemma}
\newtheorem{proposition}[theorem]{Proposition}

\newenvironment{proof}[1][Proof]{\begin{trivlist}
\item[\hskip \labelsep {\bfseries #1}]}{\end{trivlist}}

\begin{document}

\title
{\LARGE
\textbf{Row-Centric Lossless Compression of Markov Images}
}

\vspace{0mm}

\author{%
{\normalsize\begin{minipage}{\linewidth}\begin{center}
\begin{tabular}{ccc}
& {\normalsize Matthew G. Reyes} and {\normalsize David L. Neuhoff} &\\
\hspace*{0.5in} & EECS Department, University of Michigan & \hspace*{0.5in}\\
& \url{mgreyes@umich.edu} and \url{neuhoff@umich.edu} &
\end{tabular}
\end{center}\end{minipage}}
}


\maketitle
\thispagestyle{empty}

\vspace{0mm}

\begin{abstract}
	Motivated by the question of whether the recently introduced Reduced Cutset Coding (RCC) \cite{reyes2010,reyes2016a} offers rate-complexity performance benefits over conventional context-based conditional coding for sources with two-dimensional Markov structure, this paper compares several row-centric coding strategies that vary in the amount of conditioning as well as whether a model or an empirical table is used in the encoding of blocks of rows. The conclusion is that, at least for sources exhibiting low-order correlations, 1-sided model-based conditional coding is superior to the method of RCC for a given constraint on complexity, and conventional context-based conditional coding is nearly as good as the 1-sided model-based coding.
\end{abstract}

\vspace{-1mm}

\section{Introduction}\label{sec:introduction}


Lossless coding of an image involves blocking (equivalently, grouping) and ordering the pixels in some way, and feeding them, together with a corresponding set of {\em coding distributions}, to an encoder, which without loss of optimality we can assume to be an {\em Arithmetic Encoder}.  
The coding distribution for a given pixel, or block of such, is conditioned on some subset of the 
pixels, referred to as its {\em context}, that have already been encoded.

This paper considers how various coding strategies effect coding rate
and complexity. Strategies considered include different ways of blocking and ordering pixels, different contexts,  and two different ways of producing coding distributions:   model-based and  empirical, \emph{i.e.}, parametric and nonparametric.


For simplicity we focus on bilevel images. To provide a well-founded testing ground with interesting
correlation structure, we focus on images produced by a simple, uniform, Ising Markov
Random Field (MRF) model \cite{baxter},
with each pixel having four neighbors ($\cal N$, $\cal E$, $\cal S$, $\cal W$), and positive, 
row-stationary edge correlations. MRF models have seen widespread application in image processing, 
in large part due to the reasonable assumption that pixels in an image are dependent 
on some small surrounding region rather than on pixels from the entire rest of the image.  In particular, the Ising model has been proposed as a model for bilevel images \cite{reyes2014} called \emph{scenic}, which are complex bilevel images, such as landscapes and portraits, having numerous black and white regions with smooth or piecewise smooth boundaries between them.
The model-based coding distributions are based explicitly on this model.
The empirical methods simply use tables of conditional frequencies.

We focus on what we call \emph{row-centric} schemes,
which are schemes in which rows are grouped into blocks and within each block, columns are sequentially encoded from left to right.
These include both the recently introduced Reduced Cutset Coding (RCC)
\cite{reyes2010, reyes2016a}, as well as conventional context-based
conditional coding, such as in
\cite{JBIG, MemonW:97, MemonNS:2000}.
It excludes coding techniques such as when image pixels are coded in Hilbert
scan order \cite{LempelZiv:1986,MemonNS:2000}.

By the Markov property, no coding scheme could attain lower
rate than the scheme that encodes each row with coding distribution equal to the row's 
conditional distribution given the previous row as context, which has rate equal to
the entropy-rate $H_\infty = {1 \over W} H(\mathbf X_{r_1} | \mathbf X_{r_0})$, where $W$ is the width of the image, and $\mathbf X_{r_0}$ and $\mathbf X_{r_1}$ denote successive rows. An equivalent row-centric scheme will sequentially encode each pixel in a given row with context equal to the pixel to the left, the one above, and all pixels to the right of the one above.  
While it is easy to say this is optimal,
it is computationally infeasible to attain this rate exactly. 
On the one hand, with model-based coding distributions, due to need for marginalizing over
the pixels below the block, the conditional distribution of one pixel
given the aforementioned context is exorbitantly complex to compute in real time, and
exorbitantly expensive to store even if it were computed in advance.
On the other hand, with 
empirical-distribution-based coding, the distribution is again exorbitantly expensive to store.
Thus, the real question is how to approach rate $H_\infty$ with computationally 
efficient coding techniques. 
With this in mind, this paper explores the merits of several row-centric strategies -- some using model-based coding distributions and some using empirically-based distributions.  From now on, 
we call these \emph{model-based} and \emph{empirically-based} schemes, respectively.

\subsection*{Model-based}

Let $G = (V,E)$ denote a grid-graph underlying the MRF.   
For the given Ising model, one can encode an $N_b \times W$ block $\mathbf X_b$
consisting of $N_b$ rows with complexity per pixel that increases exponentially with $N_b$ 
and with storage that increases exponentially with $N_b$ and linearly with $W$.  
This is done by lumping the $i$-th column $\bfX_{b,i}$ into one super-pixel, and computing the coding distribution of each column in turn using Belief Propagation on the resulting line graph.  This is feasible for moderate $N_b$, e.g., 10 or so, and as described below, such coding distributions can be computed with conditioning/context from the row above, the row below, both the row above and row below, or from neither, without any appreciable increase in complexity.

If the coding distributions within a block are conditioned on just the row above, then to avoid an exorbitantly complex marginalization, all edges running South from the block must be cut.  This means that the computed coding distribution $p_C( \mathbf x_{b,i} )$ will not be
the true conditional distribution for the $i$-th column -- the result being that the overall coding rate will be larger than $H_\infty$ due to the divergences between the true and computed conditional distributions for the columns. Similarly, if the coding distributions within a block are computed without conditioning on either the row above or the row below, then all edges running both South and North from the block must be cut. This again means that the computed coding distributions $p_C(\mathbf x_{b,i} )$ will not be the true distributions for the columns -- the result being that the overall coding rate will exceed $H_\infty$ due both to the divergences between true and computed distributions for the columns, and the blocks being encoded independently of one another. We refer to these methods as 1-sided and 0-sided model-based coding, respectively. In each of these, the excess rate, i.e., redundancy, decreases as $N_b$ increases, and for each of these, the divergence can be minimized by choosing an appropriate {\em moment-matching} correlation for the truncated model.

Two-sided model-based coding of a block of rows is also possible, but unlike 0- and 1-sided coding, this cannot be applied to the entire image.  For example in RCC, blocks are alternately 0-sided coded and 2-sided coded.  On the one hand, the blocks that are 0-sided coded suffer the sources of redundancy mentioned previously.  On other hand, the blocks that are 2-sided coded are coded precisely at rate ${1 \over W N_b} H(\mathbf X_b | \mathbf X_S, \mathbf X_N)$, where $X_S$ and $X_N$ denote the rows just North and just South of $X_b$, respectively. While this was called RCC in \cite{reyes2010,reyes2016a}, here we refer to it as 0/2-sided coding.

\subsection*{Empirically-based}

With empirically-based coding, there could again be 0-, 1- or 2-sided coding.
However, in this paper we only consider 1-sided coding, where the
pixels in a row are sequentially coded from left to right with context consisting
of the pixel to the left and some number of pixels in the row above, beginning
with the pixel directly above and extending some number of pixels to the right.
(This is conventional context-based coding.)
While $H_{\infty}$ could be attained if all pixels to the right of the current pixel
in the row above were in  the context,  
the storage required for the empirical coding distribution increases exponentially
with the size of the context, so the size of the context must be limited to a moderate amount, for example 10.
And assuming a sufficient amount of training data that the empirical conditional distributions are
very close to the true conditional distributions, the resulting redundancy is the 
average of the divergences of the true conditional distribution of a pixel given
all values on the previous and the true conditional distribution given the moderately sized context.

\subsection*{Summary of main results}

In regard to trying to attain $H_{\infty}$ with 1-sided row-centric coding, we note that empirically-based coding uses a true distribution with a truncated context, whereas model-based coding uses an approximate distribution with full context. Moreover, 1-sided model-based coding uses an approximate distribution on all blocks, while the 0/2-sided coding of RCC uses a more severe approximation on half the blocks and an optimal distribution on the other half. Consequently, we are interested in the relative performances of these three approaches in achieving rate as close to $H_\infty$ as possible.

In this paper, we first compare 0/2-sided model-based coding with 1-sided model-based coding, and then 1-sided model-based coding with 1-sided empirical-based coding. 1-sided model-based coding has rate decreasing monotonically with $N_b$. For a given complexity, i.e.,  $N_b$, 1-sided model-based coding outperforms 0/2-sided model-based coding. Moreover, 1-sided model-based coding outperforms 1-sided empirical-based coding, though not by much. In summary, at least for Markov models exhibiting low-order correlations, there are both model-based and empirically-based 1-sided schemes with good performance and low complexity. 

The remainder of the paper is organized as follows. In Section \ref{sec:background} we cover background on the Ising model, Arithmetic Encoding, model- and empirical-based coding distributions and Reduced Cutset Coding. In Section \ref{sec:redund}, we discuss 0-, 1-, and 2-sided coding, and in Section \ref{sec:simulation} we discuss numerical results.

\section{Background}\label{sec:background}

In this section we introduce notation and background concepts and results.

\subsection{MRF Source Model}

The specific information source that we consider in the present paper is a uniform Ising model on a square grid graph $G=(V,E)$, whose nodes $V$ are the sites of an $M\times W$ rectangular lattice and whose edges $E$ are pairs of horizontally and vertically adjacent nodes. The random variable $X_i$ associated with each node $i$ assumes values in the alphabet $\calX=\{-1,1\}$ and a configuration $\bfx = (x_i:i\in V)$ has probability
\begin{eqnarray}
p(\bfx;\theta)
& = & \exp\{~ \theta \!\!\!\! \sum\limits_{\{i,j\}\in E}x_ix_j - \Phi(\theta)\},\label{eq:mrf_1}
\end{eqnarray}
where  $\Phi(\theta)$ is the log-partition function and $\theta > 0$ is the positive edge correlation parameter of the model.

\subsection{Row-Centric Arithmetic Coding}

As mentioned in the introduction, in row-centric coding, rows are grouped into $N_b\times W$ blocks and then within a block $\bfX_b$, columns of pixels are encoded from left to right. Let $r_1$ and $r_{N_b}$ denote the first and last rows, respectively, of a block. Similarly, let $r_0$ and $r_{N_b+1}$ indicate, respectively, the row preceding and row succeeding the block.   

When coding column configuration $\bfx_{b,i}$, a {\em coding distribution} 
is passed, together with the configuration $\bfx_{b,i}$, 
to an Arithmetic Encoder. In 0-sided coding $p_C(\bfx_{b,i})$ is conditioned only on $\bfx_{b,i-1}$, the configuration of the the $i-1$-st column of the block. In 1-sided model-based coding, $p_C(\bfx_{b,i})$ is conditioned on $\bfx_{b,i-1}$ and $\bfx_{r_0,i:W}$, the $i$-th through final pixels of the previous row. In 1-sided empirical-based coding, $p_C(\bfx_{b,i})$ is conditioned on $\bfx_{b,i-1}$ and $\bfx_{r_0,i:i+c-2}$, the $i$-th through $i+c-2$-th pixels of the previous row, where $c$ is the size of the context. In 2-sided model-based coding, $p_C(\bfx_{b,i})$ is conditioned on $\bfx_{b,i-1}$, $\bfx_{r_0,i:W}$, and $\bfx_{r_{b+1},i:W}$, the $i$-th through final pixels of the next row. The contexts for these schemes can be visualized with Figure \ref{fig:context}.

The approximate number of bits produced by the AC encoder when encoding the $i$-th column is
$-\log p_C(\bfx_{b_i})$. The {\em rate} $R_{b,i}$ of encoding the $i$-th column of block $b$ is the expected number of bits produced, divided by $N_b$. If the $p(\bfx_{b_i})$ is the true (conditional) distribution of column $i$ given the context, then the rate of encoding the $i$-th column is
$$R_{b,i} = \frac{1}{N_b}\left[ H(\bfX_{b,i} | C_{b,i}) + D( p(\bfx_{b,i}) || p_C(\bfx_{b,i})) \right].$$
\noindent where $\overline D$ denote divergence. From this, the rate of encoding block $b$ is 
$$ R_{b} = \frac{1}{W N_b}\left[ H(\bfX_{b} | C_b) + \overline D \right], $$
\noindent where $\overline D$ is the sum of the per-column divergences.


\subsection{Model and empirical based coding distributions}

For model-based methods, the coding distribution is computed by running BP on the Ising model restricted to the subgraph induced by the block of rows, with a possibly modified correlation parameter. In the 0- and 1-sided cases, the edge correlation parameter is adjusted to account for the truncated edges (on both sides of the block or below the block, respectively). In the case of 1- and 2-sided coding, in which conditioning on either the upper or both the upper and lower boundaries is part of the coding distribution, this conditioning is incorporated by introducing {\em self correlation} on the bottom and top rows of the block that bias those sites toward the value of their boundary neighbor. 

Let $\theta^*_{0,N_b}$ and $\theta^*_{1,N_b}$ denote the parameters used for encoding a block with 0-, respectively, 1-sided coding. For 2-sided, the block is encoded using the original parameter $\theta$, and the model becomes

\vspace{2mm}

$p(\bfx_{b} | \bfx_{r_0},\bfx_{r_{b+1}};\theta^*_{2,N_b}) = $

\vspace{.5mm}

$$\exp\{~ \theta^*_{2,N_b} \!\!\!\! \sum\limits_{\{i,j\}\in E_b} \!\!\!\! x_ix_j + \theta^*_{2,N_b} \!\!\!\!  \sum\limits_{\{i\}\in r_1\cup r_b} \!\!\!\! s_ix_i - \Phi(\theta^*_{2,N_b})\},$$

\noindent where $E_b$ is the set of edges both of whose endpoints are in $b$, and $s_i$ is the self-correlation on pixel $i$ corresponding to the value of its neighbor on the boundary of $b$.

For 1-sided coding, the model is

\vspace{2mm}

$p(\bfx_{b} | \bfx_{r_0};\theta^*_{1,N_b}) = $

\vspace{.5mm}

$$\exp\{~ \theta^*_{1,N_b} \!\!\!\! \sum\limits_{\{i,j\}\in E_b} \!\! x_ix_j + \theta^*_{1,N_b}  \sum\limits_{i\in r_1}  s_ix_i - \Phi(\theta^*_{1,N_b})\}$$

For 0-sided coding, the model is

\begin{eqnarray}
p(\bfx_{b};\theta^*_{0,N_b})
& = & \exp\{~ \theta^*_{0,N_b} \!\!\!\! \sum\limits_{\{i,j\}\in E_b} \!\!\!\! x_ix_j - \Phi(\theta^*_{0,N_b})\}, \nonumber \label{eq:mrf_0}
\end{eqnarray}

In each of these cases, the coding distribution $p(\bfx_{b,i})$ for the $i$-th column within the block is computed using Belief Propagation \cite{reyes2016a}. 
Messages are first passed from right to left on the resulting line-graph of superpixels (columns) in such a way that after the messages are received at the first column, encoding can proceed from left to right with the coding distributions being computed as they are needed. 
The (column) coding distributions for 0-, 1-, and 2-sided model-based coding are denoted $p(\bfx_{b,i} | \bfx_{b,i-1}; \theta^*_0 )$, $p(\bfx_{b,i} | \bfx_{b,i-1}, \bfx_{r_0,i:W} ; \theta^*_1)$, and $p(\bfx_{b,i} | \bfx_{b,i-1}, \bfx_{r_0,i:W}, \bfx_{r_{b+1},i:W} ; \theta^*_2)$, respectively.


%
%
%
%

Empirical coding distributions are based on a table of the frequencies of different configurations of a column for all possible configurations of the context. Letting $\bfx_T$ denote the configuration being encoded and $\bfx_C$ denote the configuration of the context, the table consists of values of the form $p^*(\bfx_T , \bfx_C)$, from which the coding distribution $p^*(\bfx_{b,i} | \bfx_{b,i-1}, \bfx_{r_0,i:i+c-2})$ can be computed, where $c$ is the size of the context.

\begin{figure}
	\vspace{-23mm}
	\begin{center}
		\hbox{
			\hspace{30mm}
			\includegraphics[scale=.37]{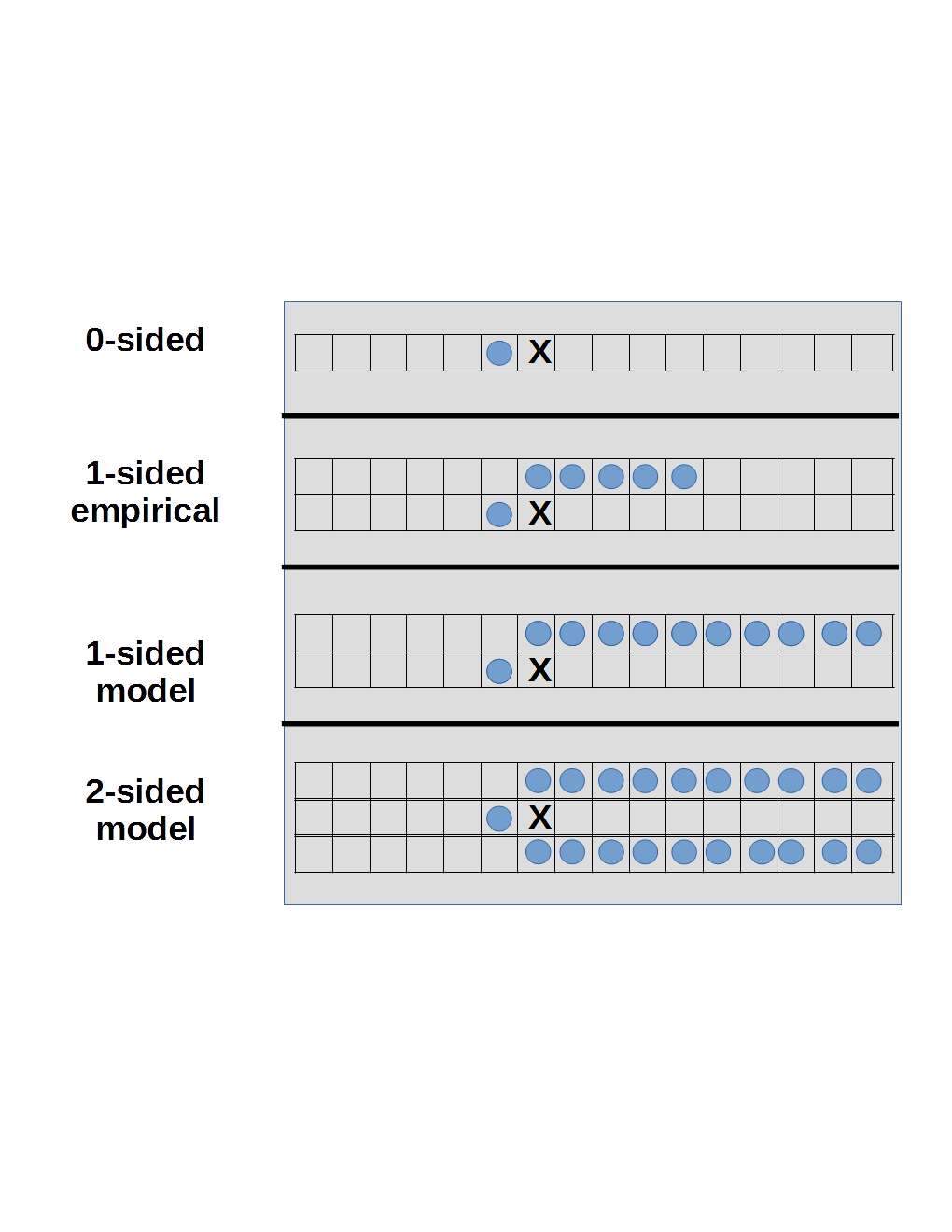}
		}
	\end{center}
	\vspace{-35mm}
	\caption{Context sets in 0-sided, 1-sided empirical-based, 1-sided model-based, and 2-sided model-based coding. The pixel being encoded is indicated with an {\bf X} while the context pixels are depicted with a blue circle.}
	\label{fig:context}
\end{figure}

There are 1-pass and 2-pass methods. In this paper we consider only the 2-pass method in which the relevant frequencies are collected from a set of training images, and then, in a second pass, the rows of the image are encoded using the collected frequencies as coding distributions.

\subsection{Reduced Cutset Coding \cite{reyes2010,reyes2016a}}
In the Reduced Cutset Coding (RCC) method introduced in \cite{reyes2010} and further analyzed in \cite{reyes2016a}, an image is divided into alternating blocks of rows $\bfX_L$ and $\bfX_S$ of sizes $N_L \times W$ and $N_S \times W$, called {\em lines} and {\em strips}, respectively. Lines are encoded first in a 0-sided manner, i.e., with no conditioning. The parameter $\theta^*_{0,N_L}$ used for the coding distributions of columns is chosen to be the one that minimizes divergence with the true distribution of lines. It is referred to as the {\em moment-matching} correlation parameter. The coding rate for lines is
\begin{eqnarray}
	R^L_{N_L} & = & \frac{1}{W N_L} \left[ H(\bfX_{L} ; \theta^*_{0,N_L} + \overline D) \right], \nonumber
\end{eqnarray}
\noindent where $\overline D$ is the divergence between $p(\bfx_b;\theta)$ and $p(\bfx_b;\theta^*_{0,N_L})$.

Strips are subsequently encoded in a 2-sided manner, i.e., conditioned on the immediately preceding and immediately succeeding rows. The coding rate for a strip is
\begin{eqnarray}
R^S_{N_S} & = & \frac{1}{W N_S}H(\bfX_{S} \mid \bfX_{r_0}, \bfX_{r_{N_{b+1}}};\theta^*_{2,N_S}). \nonumber
\end{eqnarray}

For a large image, the overall rate of RCC is then
\begin{eqnarray}  \label{eq:RSL}
R_{N_S,N_L} & \approx & \frac{N_S}{N_S + N_L}R^S_{N_S} + \frac{N_L}{N_S + N_L}R^{L}_{N_L} \nonumber \\
			& \approx & H_\infty \nonumber \\
			&         & + \frac{N_L}{N_L + N_S}\overline D + \frac{N_S}{N_L + N_S}I(\bfX_{r_0} ; \bfX_{r_{N_L+1}}) \nonumber
\end{eqnarray}
\noindent where $\overline D$ is the divergence between $p(\bfx_b;\theta)$ and $p(\bfx_b;\theta^*_{0,N_L})$, and $I(\bfX_{r_0} ; \bfX_{r_{N_L+1}})$ is the information between the row immediately preceding and the row immediately following a strip.

%

\vspace{2mm}

\vspace{2mm}

\section{Row-Centric Coding Redundancy}\label{sec:redund}

In this section we return to the question posed in Section \ref{sec:introduction}, that of attaining rate as close as possible to the entropy rate $H_\infty = H(\bfX_{r_1} | \bfX_{r_0})$, and discuss the redundancies associated with the different coding strategies considered in this paper. While we cannot analytically evaluate the rate of decrease of the redundancies, by performing numerical experiments as in the next section, we can gain a sense of the relative rates of decrease.

We let $R^{0E}_{N_b}$ and $R^{0M}_{N_b}$ denote the rate for coding ${N_b}$ rows with 0-sided empirical- and model-based coding, respectively. Likewise for $R^{1E}_{N_b}$, $R^{1M}_{N_b}$, and $R^{2M}_{N_b}$. We focus here on the coding of a single row, i.e., $N_b = 1$. Moreover, let $I(\bfX_{r_1} ; \bfX_{r_0})$ be the mutual information between rows 0 and 1. Some of the results in this section make use of Lemma \ref{lemma:div} in Section \ref{sec:append}.

\vspace{2mm}

\begin{proposition}

The rate for encoding a row with 0-sided model-based coding is 

$$R^{0M}_1 = H_\infty + \frac{1}{W }\left[ \overline D^{0M}_1 + I(\bfX_{r_1} ; \bfX_{r_0}) \right] $$

\noindent where $\overline D^{0M}_1$ is the sum of divergences between $p(\bfx_{b,i} | \bfx_{b,i-1}; \theta )$ and $p(\bfx_{b,i} | \bfx_{b,i-1}; \theta^*_0 )$ over all columns. 

\end{proposition}

\vspace{2mm}

\begin{proof}
	
	\begin{eqnarray}
		R^{0M}_1 & = & \frac{1}{W}\left[ H(X_{b}) + D(X_{b} || \tilde X_{b}) \right] \nonumber \\
				 & = & \frac{1}{W}\left[ H(X_{b} | X_{r_0}) + I(X_{r_1} ; X_{r_0}) + D(X_{b} || \tilde X_{b}) \right] \nonumber \\
				 & = & H_{\infty} + \frac{1}{W}\left[ I(X_{r_1} ; X_{r_0}) + D(X_{b} || \tilde X_{b}) \right], \nonumber
	\end{eqnarray}
	\noindent which shows the proposition. $\hfill\Box$
\end{proof}

\vspace{2mm}

\begin{proposition}

The rate for coding a row with 0-sided empirical-based coding is 

$$R^{0E}_1 = H_\infty + \frac{1}{W }\left[ I(\bfX_{r_1} ; \bfX_{r_0}) \right].$$

\end{proposition}

\vspace{2mm}

\begin{proof}
	\begin{eqnarray}
		R^{0E}_1 & = & \frac{1}{W} H(\bfX_{b}) \nonumber \\
		& = & \frac{1}{W}\left[ H(\bfX_{b} | \bfX_{r_0}) + I(\bfX_{r_1} ; \bfX_{r_0})\right] \nonumber \\
		& = & H_{\infty} + \frac{1}{W} I(\bfX_{r_1} ; \bfX_{r_0}), \nonumber
	\end{eqnarray}
	\noindent which shows the proposition. $\hfill\Box$
\end{proof}

\vspace{2mm}

Note that both 0-sided methods suffer the information penalty for independently encoding rows of the image. However, we do not include a divergence term in $R^{0E}_1$ because given enough training data, the empirical coding distribution $p^*(\bfx_{b,i} | \bfx_{b,i-1})$ for the $i$-th column will well-approximate the true distribution $p(\bfx_{b,i} | \bfx_{b,i-1}; \theta )$. Thus one could estimate $\bar D^{0M}_{N_b}$ by encoding the source with both 0-sided model-based coding and 0-sided empirical-based coding and forming the estimate $R^{0M}_{N_b} - R^{0E}_{N_b}$.


\begin{proposition}

The rate for coding a row with 2-sided model-based coding is

\begin{eqnarray}
	R^{2M}_1 & = & H_\infty - \frac{1}{W}I(\bfX_{r_1} ; \bfX_{r_2} | \bfX_{r_0}) \nonumber \\
			 & < & H_\infty, \nonumber
\end{eqnarray}

\end{proposition}

\vspace{2mm}

\begin{proof}
	\begin{eqnarray}
		R^{2M}_1 & = & \frac{1}{W}H(\bfX_b | \bfX_{r_0},\bfX_{r_2}) \nonumber\\
				 & = & \frac{1}{W}\left[ H(\bfX_b | \bfX_{r_0}) - I(\bfX_b ; \bfX_{r_w} | \bfX_{r_0}) \right] \nonumber \\
				 & = & H_{\infty} + \frac{1}{W}I(\bfX_{r_1} ; \bfX_{r_2} | \bfX_{r_0}) \nonumber
	\end{eqnarray}
\end{proof}

This, of course, is not an actual coding rate, but it can be shown that when combined with $R^{0M}_1$ gives the performance of RCC with $N_L = N_S = 1$.

\vspace{2mm}

\begin{proposition}
	Encoding every other row with 0-sided model-based coding and 2-sided model-based coding gives rate
	\begin{eqnarray}
		\frac{1}{2} \left[ R^{0M}_1 + R^{2M}_1 \right] & = & H_{\infty} + \frac{1}{2 W}\bar D^{0M}_1 + \frac{1}{2 W} I(\bfX_{r_2} ;\bfX_{r_0}) \nonumber
	\end{eqnarray}
\end{proposition}

\begin{proof}
	\begin{eqnarray}
		\frac{1}{2} \left[ R^{0M}_1 + R^{2M}_1 \right] & = & \frac{1}{2} H_{\infty} + \frac{1}{2 W}\left[ I(\bfX_{r_1} ; \bfX_{r_0}) + D(\bfX_{b} || \tilde \bfX_{b}) \right] + \frac{1}{2} \left[ H_\infty - \frac{1}{W}I(\bfX_{r_1} ; \bfX_{r_2} | \bfX_{r_0}) \right] \nonumber \\
			& = & H_{\infty} + \frac{1}{2 W} \bar D^{0M}_1 + \frac{1}{2 W} \left[ I(\bfX_{r_1} ; \bfX_{r_0}) - I(\bfX_{r_1};\bfX_{r_2} | \bfX_{r_0}) \right]. \nonumber
	\end{eqnarray}
	
	Therefore, to show the proposition we need to show that $I(\bfX_{r_1} ; \bfX_{r_0}) - I(\bfX_{r_1};\bfX_{r_2} | \bfX_{r_0}) = I(\bfX_{r_2} ; \bfX_{r_0})$. To do this, we note that under a row stationary Markov model such as the one considered in this paper, we have
	
	\begin{eqnarray}
		I(\bfX_{r_1} ; \bfX_{r_0}) - I(\bfX_{r_1};\bfX_{r_2} | \bfX_{r_0}) & = & H(\bfX_{r_1}) - H(\bfX_{r_1} | \bfX_{r_0}) - H(\bfX_{r_2} | \bfX_{r_0}) + H(\bfX_{r_2} | \bfX_{r_0}, \bfX_{r_1}) \nonumber \\
				& = & H(\bfX_{r_1}) - H(\bfX_{r_1} | \bfX_{r_0}) - H(\bfX_{r_2} | \bfX_{r_0}) + H(\bfX_{r_2} | \bfX_{r_1}) \label{eq:markov} \\
				& = & H(\bfX_{r_2}) - H(\bfX_{r_2} | \bfX_{r_1}) - H(\bfX_{r_2} | \bfX_{r_0}) + H(\bfX_{r_2} | \bfX_{r_1}) \label{eq:stationary} \\
				& = & H(\bfX_{r_2}) - H(\bfX_{r_2} | \bfX_{r_0}) \nonumber \\
				& = & I(\bfX_{r_2} ; \bfX_{r_0}) \nonumber
	\end{eqnarray}
	\noindent where (\ref{eq:markov}) is from the Markov property and (\ref{eq:stationary}) is from row stationarity. This completes the proof. $\hfill\Box$
\end{proof}

By estimating $\bar D^{0M}_{N_b}$ using the rates $R^{0M}_{N_b}$ and $R^{0E}_{N_b}$ from 0-sided model-based and 0-sided empirical-based coding, we can then subtract this from the rate of RCC and obtain an estimate of the shape of $I(\bfX_{r_0} ; \bfX_{r_{N_b + 1}})$.

Using the above notation, we can restate Proposition 3.1 of \cite{reyes2016a}, for all $N_0$ and $N_2$, as

\vspace{2mm}

\begin{proposition}{RCC}
$$ R^{0M}_{N_0+1} < R^{0M}_{N_0}, ~~~~~~ R^{2M}_{N_2+1} > R^{2M}_{N_2}, ~~~~~~ R^{0M}_{N_0} > R^{2M}_{N_2}.$$
\end{proposition}

\vspace{2mm}

\begin{proof}
	The proofs can be found in \cite{reyes2016b}.
\end{proof}

\vspace{2mm}

We now consider rates of 1-sided coding.

\vspace{2mm}

\begin{proposition}

The rate for encoding a row with 1-sided model-based coding is

$$R^{1M}_1 = H_\infty + \frac{1}{W}\overline D^{1M}_1$$

\noindent where $\overline D^{1M}_1$ is the sum of divergences between $p(\bfx_{b,i} | \bfx_{b,i-1}, \bfx_{r_0,i:W} ; \theta)$ and $p(\bfx_{b,i} | \bfx_{b,i-1}, \bfx_{r_0,i:W} ; \theta^*_1)$ over all columns. 

\end{proposition}

\vspace{2mm}

\begin{proof}
	\begin{eqnarray}
		R^{1M}_1 & = & \frac{1}{W}\left[ H(X_{r_1} | X_{r_0}) + D( X_{r_1} | X_{r_0} || \tilde X_{r_1} | X_{r_0}) \right], \nonumber
	\end{eqnarray}
	\noindent where $D( X_{r_1} | X_{r_0} || \tilde X_{r_1} | X_{r_0})$ is the divergence between the true conditional distribution of a row conditioned on the previous row and the conditional distribution of a row conditioned on the previous row using the 1-sided model, which can be expressed as the sum of divergences between $p(\bfx_{b,i} | \bfx_{b,i-1}, \bfx_{r_0,i:W} ; \theta)$ and $p(\bfx_{b,i} | \bfx_{b,i-1}, \bfx_{r_0,i:W} ; \theta^*_1)$. This shows the proposition. $\hfill\Box$
\end{proof}

Similarly, the rate of encoding a row with 1-sided empirical-based coding is


\vspace{2mm}

\begin{proposition}
	
The rate for encoding a row with 1-sided empirical-based coding is

$$R^{1E}_1 = H_\infty + \frac{1}{W}\overline D^{1E}_1$$

\noindent where $\overline D^{1E}_1$ is the sum of divergences between $p(\bfx_{b,i} | \bfx_{b,i-1}, \bfx_{r_0,i:W} ; \theta)$ and $p^*(\bfx_{b,i} | \bfx_{b,i-1}, \bfx_{r_0,i:i+c-2})$ over all columns. 

\end{proposition}

\vspace{2mm}

\begin{proof}
	\begin{eqnarray}
	R^{1M}_1 & = & \frac{1}{W}\left[ H(X_{r_1} | X_{r_0}) + D( X_{r_1} | X_{r_0} || \tilde X_{r_1} | X_{r_0}) \right], \nonumber
	\end{eqnarray}
	\noindent where $D( X_{r_1} | X_{r_0} || \tilde X_{r_1} | X_{r_0})$ is the divergence between the true conditional distribution of a row conditioned on the previous row and the conditional distribution of a row conditioned on the previous row using the 1-sided empirical distributions, which can be expressed as the sum of divergences between $p(\bfx_{b,i} | \bfx_{b,i-1}, \bfx_{r_0,i:W} ; \theta)$ and $p^*(\bfx_{b,i} | \bfx_{b,i-1}, \bfx_{r_0,i:i+c-2})$. This shows the proposition. $\hfill\Box$
\end{proof}

\vspace{2mm}

Note that the two 1-sided coding scemes do not suffer an explicit information penalty because there is conditioning on the previous row. On the other hand, if the context size $c$ could be chosen as $c = W + 2 - i$ for each column $i$, then the divergence term $\overline D^{1E}_1$ would vanish. Thus $\overline D^{1E}_1$ is really a sum of conditional information terms. However, both $\overline D^{1M}_1$ and $\overline D^{1E}_1$ are less than $\overline D^{0M}_1$, so it is of interest how these smaller divergences on all blocks compare with the 0/2-sided scheme of RCC in which half the blocks have a larger divergence, plus an information penalty, while the other half actually receive a coding rate reduction.

%
%
%
%
%

Analogous to the results of \cite{reyes2016a}, 1-sided model-based coding can be shown to have the following properties.

\vspace{2mm}

\begin{proposition}\label{prop:1sided}
	For all $N_b$ and $N_2$,
	$$ R^{1M}_{N_b+1} < R^{1M}_{N_b} ~~~~~~ R^{1M}_{N_b} < R^{0M}_{N_b} ~~~~~~ R^{1M}_{N_b} > R^{2M}_{N_2} $$	
\end{proposition}




%
%
%
%
%
%
%

\section{Numerical Results and Comparisons}\label{sec:simulation}

Using Gibbs sampling, we generated configurations $\bfx^{(1)},\ldots,\bfx^{(17)}$ of a $200\times 200$ modeled by an Ising MRF with $\theta = .4$. On this dataset we tested three strategies: 0/2-sided model-based coding, 1-sided model-based coding, and 1-sided empirical-based coding. The estimates $\theta^*_{0,n}$, $\theta^*_{1,n}$, and $\theta^*_{2,n}$ were found as in \cite{reyes2016a} and are shown in Figure \ref{fig:params}. 

Figures \ref{fig:modelrates} and \ref{fig:1sidedrates} show the rates attained by the various row-centric coding schemes considered in this paper, as a function of block size parameter $n$. These rates  were computed  by averaging the negative logarithm of the coding distributions evaluated at the actual pixel/super-pixel values. In \cite{reyes2016a} we observed that for a given complexity, i.e., given the maximum of $N_L$ and $N_S$, the best performance of 0/2-sided coding was found when lines and strips have the same size, i.e., $N_L = N_S = N_b$. Thus in the model-based comparison, our 0/2-sided method uses lines and strips of equal height. 

%
%
%
%
%
%
%

As predicted by Proposition \ref{prop:1sided}, Figure \ref{fig:modelrates} shows that $R^{1M}_{N_b}$ is decreasing in $N_b$, $R^{1M}_{N_b} < R^{0M}_{N_b}$ and $R^{1M}_{N_b} > R^{2M}_{N'_b}$ for all $N_b$ and $N'_b$. Also in Figure \ref{fig:modelrates}, we observe that for a given block size $N_b$, 1-sided model-based coding achieves lower rate than 0/2-sided model-based coding. Indeed, 1-sided model-based coding with $N_b=1$ nearly as good as 0/2-sided coding with $N_b=7$. Moreover, using the 2-sided coding rate as a lower bound for $H_\infty$, we can say that with $N_b = 3$, 1-sided model-based coding comes to within 3.5\% of $H_\infty$.

%
%

\begin{figure}
	\centerline{    \hbox{			
			\includegraphics[scale = .45]{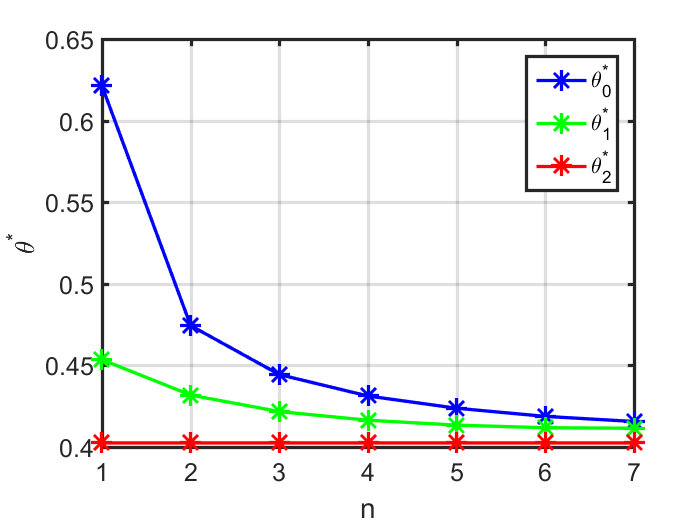}
	}   }
	\caption{ Parameters used for 0-, 1-, and 2-sided model-based coding.}
	\label{fig:params}
\end{figure}

Figure \ref{fig:1sidedrates} shows the rate of 1-sided model-based coding with $N_b = 1$, and 1-sided empirical-based coding for varying sizes of context. Note that context size $c = 1$ actually corresponds to 0-sided empirical-based coding, since in this scheme, only the pixel to the left is used as context. We observe that 1-sided model-based coding with $N_b=1$ achieves lower rate than 1-sided empirically-based coding with all context sizes we considered. The difference between the rates of 1-sided model-based and 1-sided empirical-based coding shrinks with context size and when the context size is 5, the difference is about .0025 bpp or .4\%. Improvements after that are very slow. Again using the rate of 2-sided model-based coding as a lower bound for $H_\infty$, we observe that 1-sided empirically-based coding with context size 5 comes with 4\% of entropy-rate.

\begin{figure}
	\centerline{    \hbox{
			\hspace{0in}
			\includegraphics[scale = .45]{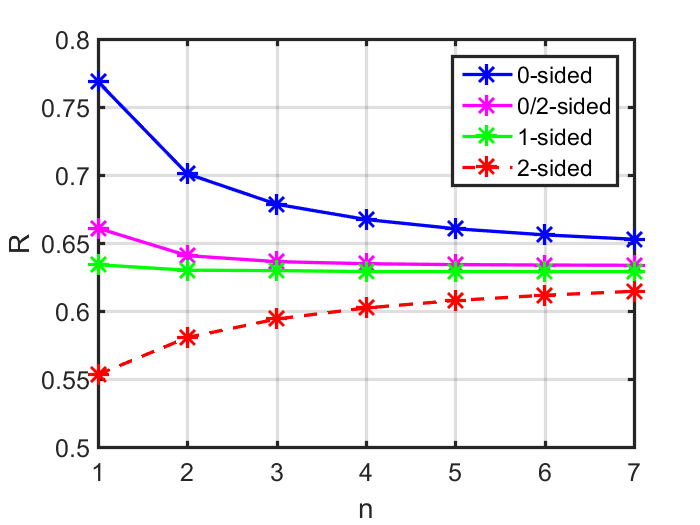}
	}   }
	\vspace{-5pt}
	\caption{0-, 0/2-, 1-, and 2-sided coding rates for model-based methods.}
	\label{fig:modelrates}
\end{figure}

\begin{figure}
	\centerline{    \hbox{
			\hspace{0in}
			\includegraphics[scale = .45]{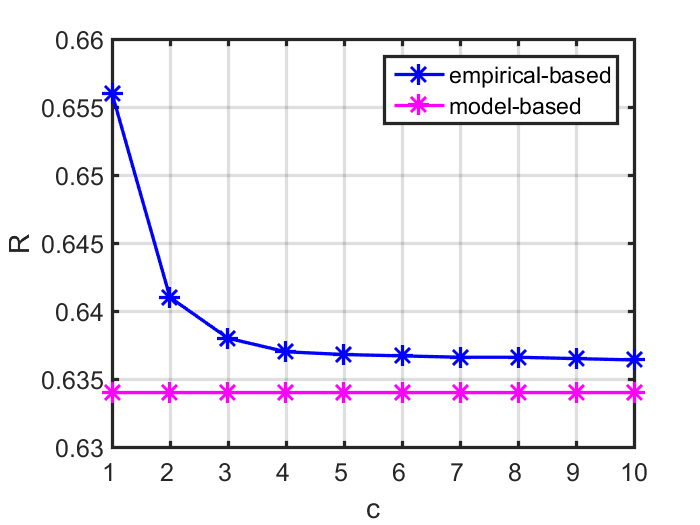}
	}   }
	\vspace{-5pt}
	
	\caption{Empirical- and model-based coding rates for 1-sided coding.}
	\label{fig:1sidedrates}
\end{figure}

Another interesting observation is made by recalling from the previous section that while both 0-sided model-based and 0-sided empirical-based coding methods suffer an information penalty, the model-based scheme suffers an additional divergence penalty $\overline D^{0M}_{N_b}$. Therefore, by comparing the $n=1$ point on the 0-sided rate curve of Figure \ref{fig:modelrates} with the $c=1$ point on the empirical-based rate curve of Figure \ref{fig:1sidedrates}, we can estimate that the normalized divergence between $p(\bfx_b;\theta)$ and $p(\bfx_b;\theta^*_0)$ for a single row is about .1 bits per pixel. Moreover, by again using the 2-sided model-based rate curve as a lower bound for $H_\infty$, we can bound the normalized information $I(\bfX_2 ; \bfX_1)$ between successive rows by .041 bits per pixel.

\section{Concluding Remarks}

In this paper we posed the problem of considering different approaches to what are called row-centric coding. We presented the problem in the context of a standard MRF image model in order to provide a well-founded testing ground in which model-based and empirical-based approaches can be compared, and moreover, 1-sided coding can be compared to the tradeoffs in 0/2-sided coding.


\newpage

\section{Appendix}\label{sec:append}

\begin{lemma}\label{lemma:div}
	For random variables $\bfX_1,\ldots,\bfX_N$, $N\geq 2$, let $p_{i|C_i}$ be the probability of $\bfX_i$ given $\bfX_{C_i}$, where $C_i\subset\{1,\ldots,i-1\}$ is the context for $\bfX_i$ and let $q_{i|\bar C_i}$ be the coding distribution for $\bfX_i$, where $\bar C_i\subset \{1,\ldots,i-1\}$ is the context for $\bfX_i$ under the $q$ distribution. Then,
	\begin{eqnarray}
		D(\prod\limits_{i}^N p_{i|C_i} || \prod\limits_{i} q_{i|\bar C_i}) & = & \sum\limits_{i=1}^N \sum\limits_{\bfx_{C_i\cup \bar C_i}}p_{C_i} D( p_{i|C_i} || q_{i| \bar C_i} ) \nonumber
	\end{eqnarray}
\end{lemma}

\vspace{2mm}

\begin{proof}
	First consider the case where $C_i = \bar C_i = \{1,\ldots,i-1\}$. We will prove it by induction. Letting $N=2$ we have that
	\begin{eqnarray}
		D( p_1p_{1|2} || q_1q_{1|2}) & = & \sum\limits_{\bfx_1 \bfx_2} p_1p_{2|1} \log \frac{p_1 p_{2|1}}{q_1 q_{2|1}} \nonumber \\	
							& = & \sum\limits_{\bfx_1, \bfx_2} p_1 p_{2|C_2} \left[ \log\frac{p_1}{q_1} + \log\frac{p_{2|C_2}}{q_{2|\bar C_2}} \right] \nonumber \\
							& = & \sum\limits_{\bfx_1, \bfx_2} p_1 p_{2|C_2} \log\frac{p_1}{q_1} + \sum\limits_{\bfx_1, \bfx_2} p_1 p_{2|C_2} \log\frac{p_{2|C_2}}{q_{2|C_2}} \nonumber \\
							& = & \sum\limits_{\bfx_1} p_1 \log\frac{p_1}{q_1} \sum\limits_{\bfx_2} p_{2|C_2} + \sum\limits_{\bfx_1} p_1 \sum\limits_{\bfx_2} p_{2|C_2} \log\frac{p_{2|C_2}}{q_{2|\bar C_2}} \nonumber \\
							& = & \sum\limits_{\bfx_1} p_1 \log\frac{p_1}{q_1} + \sum\limits_{\bfx_1} p_1 \sum\limits_{\bfx_2} p_{2|C_2} \log\frac{p_{2|C_2}}{q_{2|\bar C_2}} \nonumber \\
							& = & \sum\limits_{i=1}^2 \sum\limits_{\bfx_1,\ldots,\bfx_{i-1}} p_{1,\ldots,i-1} \sum\limits_{\bfx_i} p_{i | C_i} \log \frac{p_{i | C_i}}{q_{i | \bar C_i}} \nonumber \\
							& = & \sum\limits_{i=1}^2 \sum\limits_{\bfx_1,\ldots,\bfx_{i-1}} p_{1,\ldots,i-1} \sum\limits_{\bfx_i} D( p_{i | C_i} || q_{i | \bar C_i}), \nonumber
	\end{eqnarray}
	\noindent which shows that the lemma holds for some $N = k\geq 2$. Now letting $N = k+1$, we see that
	\begin{eqnarray}
		D(\prod\limits_{i}^{k+1} p_{i|C_i} || \prod\limits_{i} q_{i|\bar C_i}) & = & \sum\limits_{\bfx_1,\ldots,\bfx_{k},\bfx_{k+1}} \prod\limits_{i=1}^{k}p_{i | C_i} p_{k+1 | C_{k+1}} \log \frac{\prod\limits_{i=1}^{k} p_{i | C_i} p_{k+1 | C_{k+1}}}{\prod\limits_{i=1}^{k} q_{i | \bar C_i} q_{k+1 | C_{k+1}}} \nonumber \\
		& = & \sum\limits_{\bfx_1,\ldots,\bfx_k} \prod\limits_{i=1}^{k} p_{i | C_i} \log \frac{\prod\limits_{i=1}^{k} p_{i | C_i}}{\prod\limits_{i=1}^{k} q_{i | \bar C_i} } \nonumber \\
		&& + \sum\limits_{\bfx_1,\ldots,\bfx_k} \prod\limits_{i=1}^{k} p_{i | C_i} \sum\limits_{\bfx_{k+1}} p_{k+1 | C_{k+1}} \log \frac{p_{k+1 | C_{k+1}}}{q_{k+1 | \bar C_{k+1}}} \nonumber \\
		& = & \sum\limits_{\bfx_1,\ldots,\bfx_k} \prod\limits_{i=1}^{k} p_{i | C_i} \log \frac{\prod\limits_{i=1}^{k} p_{i | C_i}}{\prod\limits_{i=1}^{k} q_{i | \bar C_i} } \nonumber \\
		&& + \sum\limits_{\bfx_1,\ldots,\bfx_k} \prod\limits_{i=1}^{k} p_{i | C_i} D(p_{k+1 | C_{k+1}} || q_{k+1 | \bar C_{k+1}}) \nonumber \\
		& = & \sum\limits_{i=1}^k \sum\limits_{\bfx_{C_i\cup \bar C_i}}p_{C_i} D( p_{i|C_i} || q_{i| \bar C_i} ) \nonumber \\
		&& + \sum\limits_{\bfx_1,\ldots,\bfx_k} \prod\limits_{i=1}^{k} p_{i | C_i} D(p_{k+1 | C_{k+1}} || q_{k+1 | \bar C_{k+1}}) \label{eq:indhyp} \\
		& = & \sum\limits_{i=1}^{k+1} \sum\limits_{\bfx_{C_i\cup \bar C_i}}p_{C_i} D( p_{i|C_i} || q_{i| \bar C_i} ) \nonumber
	\end{eqnarray}
\end{proof}


\section*{References}

\begin{thebibliography}{1}










\bibitem{reyes2010}
M.G. Reyes and D.L. Neuhoff, ``Lossless Reduced Cutset Coding of Markov Random Fields", DCC, Snowbird, UT, 2010.

\bibitem{reyes2016a}
M. G. Reyes and D. L. Neuhoff, ``Cutset Width and Spacing for Reduced Cutset Coding of Markov Random Fields," ISIT 2016, July 2016.

\bibitem{baxter}
R.J. Baxter, {\em Exactly Solved Models in Statistical Mechanics},
New York: Academic, 1982.

\bibitem{reyes2014}
M. G. Reyes, D. L. Neuhoff, T. N. Pappas, ``Lossy
Cutset Coding of Bilevel Images Based on Markov Random Fields," \emph{IEEE Trans. Img. Proc.}, vol. 23, pp. 1652-1665, April 2014.

\bibitem{JBIG} 
``Progressive bi-level image compression,” ISO/IEC
Int.\,Std.\,11544, 1993. 

\bibitem{MemonW:97}
N.~Memon and X.~Wu, ``Recent developments in
context-based predictive techniques
for lossless image compression," 
\emph{The Computer J.}, vol.~40, no.~2, pp.~127-136, 1997.

\bibitem{MemonNS:2000}
N.~Memon, D.L.~Neuhoff, and S.~Shende, ``An analysis of some common scanning techniques for lossless image coding," {\em IEEE Trans.~Image Proc.}, vol.~9, no.~11, pp.~1837-1848, 2000.















\bibitem{LempelZiv:1986}
A.~Lempel and J.~Ziv, ``Compression of two-dimensional data,” {\em IEEE
	Trans.~Inform.~Theory}, vol.~IT-32, no.~1, pp.~1–8, 1986.


\bibitem{reyes2016b}
M. G. Reyes and D. L. Neuhoff, ``Cutset Width and Spacing for Reduced Cutset Coding of Markov Random Fields," available online at http://arxiv.org/abs/1602.04835.








\end{thebibliography}
\end{document}